\documentclass[manyauthors]{fundam}

\publyear{22}
\papernumber{2102}
\volume{185}
\issue{1}

\usepackage{amsmath}
\usepackage{amsfonts}
\usepackage{amssymb}

\usepackage{graphicx}
\usepackage{textcomp}

\usepackage{multirow}
\usepackage{xcolor}
\usepackage{soul}
\usepackage{hyperref}
\hypersetup{
    colorlinks=true,
    linkcolor=blue,
    filecolor=magenta,      
    urlcolor=cyan,
    pdftitle={Overleaf Example},
    pdfpagemode=FullScreen,
    }

\begin{document}

\title{McEliece cryptosystem based on Plotkin construction with QC-MDPC and QC-LDPC codes. }

\author{Belkacem Imine \\
Laboratory of Coding and Security of Information\\
University of Science and Technology of\\ Oran (USTO), Algeria\\
belkacem.imine{@}univ-usto.dz
\and Naima Hadj-Said\\
Laboratory of Coding and Security of Information\\
University of Science and Technology of\\ Oran (USTO), Algeria\\
naima.hadjsaid@univ-usto.dz
\and Adda Ali-Pacha\\
Laboratory of Coding and Security of Information\\
University of Science and Technology of\\ Oran (USTO), Algeria\\
adda.alipacha@univ-usto.dz
 } 
\maketitle
\address{Laboratory of Coding and Security of Information,
University of Science and Technology of Oran (USTO),
BP 1505 El M’Naouer Oran 31036, Algeria. belkacem.imine@univ-usto.dz}
\runninghead{B. Imine, N. Hadj-Said, A. Ali-Pacha}{Plotkin-based McEliece cryptosystem with QC-MDPC and QC-LDPC.  }

\begin{abstract}
In this paper, we propose a new variant of the McEliece cryptosystem using two families of quasi-cyclic codes: low density parity check codes (QC-LDPC) and moderate density parity check codes (QC-MDPC).
Due to the low weight codewords in the dual of LDPC codes, this family of codes is vulnerable to dual code attacks, making it unsuitable for use with the McEliece cryptosystem.
However, this is not the case in our proposal, and it is possible by using the $(U|U+V)$ construction to concatenate LDPC codes with MDPC codes.
We will demonstrate that our proposed cryptosystem can withstand dual code and generic decoding attacks, and that the public key can be reduced by leveraging the quasi-cyclic property and the Plotkin construction. \end{abstract}


%

\section{Introduction}
Cryptography's goal is to make any conversation illegible to everyone except the intended recipients, and in order to do so, complex problems like the hardness of factoring a large integer problem, which is used in the RSA cryptosystem \cite{rivest1978method}, and the discrete logarithm problem, which is used in~\cite{diffie1976new}, must be addressed.
Quantum computers, on the other hand, can solve both of these problems, rendering any cryptographic system that relies on them insecure.
R.J.McEliece proposed a public-key encryption scheme based on the difficulty of decoding a random code in 1978~\cite{mceliece1978public}.

McEliece's cryptosystem is best suited for future applications as an alternative to cryptosystems that rely on number theory problems because decoding any large code is an NP-complete problem~\cite{berlekamp1978inherent}.
At first, it was suggested that binary Goppa codes be used as the underlying code in McEliece's cryptosystem~\cite{mceliece1978public}, but the main disadvantage is that the public key size is large, making it unusable in practice.
After several attempts to reduce the size of the public key, H. Niederreiter proposed the dual code as a public key~\cite{niederreiter1986knapsack}.
In~\cite{sidelnikov1994public}, V.M. Sidelnikov proposed a new variant of the McEliece cryptosystem based on algebraic-geometric codes, replacing binary Goppa codes with binary Reed-Muller codes, and H.Janwa and Moreno proposed a new variant of the McEliece cryptosystem based on algebraic-geometric codes, replacing binary Goppa codes with binary Reed-Muller codes~\cite{janwa1996mceliece}. 

Monico et al. proposed the use of LDPC codes in~\cite{monico2000using}.
Gaborit proposed in~\cite{Gaborit2005Shorter} a method for reducing key size by utilizing the (QC) property of a linear code.
The QC-LDPC codes~\cite{baldi2007qc} were proposed by Baldi et al., and the QC-MDPC codes were proposed by Misoczki et al. in \cite{misoczki2013mdpc}.
With the exception of~\cite{mceliece1978public} and~\cite{misoczki2013mdpc}, none of these attempts were successful. 

Moufek et al. proposed a technique in~\cite{moufek2016new} that involved juxtaposing QC-LDPC and QC-MDPC codes, but an efficient attack was described in~\cite{dragoi2017cryptanalysis}.
The idea behind the~\cite{dragoi2017cryptanalysis} attack is to build the parity-check matrix of the public QC-LDPC code because the dual code contains very low-weight codewords and finding them using information set decoding attack~\cite{stern1988method} or~\cite{leon1988probabilistic} is not difficult, then find an alternative permutation matrix that allows to construct an equivalent QC-LDPC code that can decode the ciphertext with a high probability.
\subsection{Motivation and contribution}
\begin{itemize}
    \item \textbf{Increasing the transmission rate of the~\cite{moufek2016new} cryptosystem}: In~\cite{moufek2016new}, Moufek et al. have used the QC-LDPC and QC-MDPC codes together in juxtaposition to achieve their results.
Given the generator matrices $G_1$ of the QC-LDPC code $C_1$ and $G_2$ of the QC-MDPC code $C_2$, with lengths $n_1$ and $n_2$, respectively, the public generator matrix is $G_p=S [G_1 | G_2] P$, where $S$ is a $k$-dimensional square, dense and non-singular matrix, and $P$ is a permutation matrix of length $n_1+n_2$. If $m$ is the message to encrypt of length $k$, then the associated ciphertext is $c=mG_p+e$, such that $e$ is a binary vector with weight $t_1$ in the first $n_1$ positions and weight $t_2$ in the last $n_2$ positions. As we can see, the transmission rate $R=k/(n_1+n_2)$ is extremely low.
However, due to the Plotkin construction, the transmission rate $R'=2k/(n_1+n_2)$ in our proposal is greater than $R$.  
\item \textbf{Enhancing the security level of the~\cite{moufek2016new} cryptosystem}: The authors in~\cite{dragoi2017cryptanalysis} have succeeded in deciphering the encrypted message $c$ by first targeting the secret QC-LDPC code and finding an equivalent parity check matrix. Finally, once an equivalent matrix for the QC-LDPC code is found, with high probability the ciphertext can be easily decrypted without targeting the MDPC code. Even though an adversary can find an equivalent parity check matrix of the QC-LDPC in our proposal, he cannot find the ciphertext. 
\item \textbf{Avoid the weak keys-based message-recovery attack}: It has been demonstrated that the LDPC code-based cryptosystem is insecure because its dual code contains a large number of low-weight codewords~\cite{baldi2007qc, dragoi2017cryptanalysis}.
However, even this undesirable feature of the QC-LDPC code can not endanger our cryptosystem. 
\end{itemize}
\subsection{Paper organization}
The remainder of the paper is organized as follows.
Section~\ref{secPriliminaries} discusses linear codes in general, how to build a QC-LDPC/QC-MDPC code, and the McEliece cryptosystem.
Section~\ref{secProp} describes our cryptosystem in detail, including key generation, encryption, and decryption.
In section~\ref{secSecAnalysis}, we examine the robustness of our cryptosystem by mentioning the most powerful key-recovery attacks against the QC-LDPC code-based and QC-MDPC code-based cryptosystems, as well as the most well-known message-recovery attacks.
We reach a conclusion in section~\ref{conc}.  
\section{PRELIMINARIES}
\label{secPriliminaries}
Let $\mathbb{F}_q$ to be a finite field of $q$ elements, and let $[n,k]_q$ denote the $k$-dimensional linear subspace $C$ of the vector space $\mathbb{F}_q^n$, then the $k$ $q$-ary linearly independent vectors set of $\mathbb{F}_q^n$ that spans $C$ is called the generator matrix of $C$ and denoted by $G$. The code orthogonal of $C$ is called the dual code of $C$ and it is generated by the generator matrix $H$ of rank $n-k$ and length $n$ such that $GH^T=0$. The matrix $H$ is called the parity check matrix of $C$.
\par
Let $v \in \mathbb{F}_q^n$. Then the hamming-weight $(wt)$ of $v$ is the number of non-zero coordinates of $v$, $wt(v)=support \{ {v_i \in \mathbb{F}_q} : v_i \neq 0 \}$. An $[n,k,d]$ code $C$ is an $[n,k]$ code which has a minimum hamming-weight of non-zero codewords $d$ and a hamming-distance between codewords is at least $d$ an it can correct up to $t=(d-1)/2$ errors. The code $C$ is systematic if its generator matrix $G$ is in the systematic form $G= 
  \left[ \begin{array}{@{}c|c@{}}
  I_{k \times k} & P_{k \times (n-k)}\\
\end{array} \right]$, then its parity check matrix $H$ has the form $H= 
  \left[ \begin{array}{@{}c|c@{}}
    -P^T & I_{(n-k) \times (n-k)}\\
\end{array} \right]$.
If there exist an integer $1<n_0<n$ such that every cyclic-shift of $n_0$ positions of any codeword results a valid codeword, then the code $C$ is called quasi-cyclic and if $n_0=1$, then the code is cyclic code.

An $[n,r,w]$-LDPC code is a linear code of length $n$ and co-dimension $r$ characterized by its parity check matrix $H$ whose rows of weight $w=\mathcal{O}(1)$~\cite{gallager1962low}. For the $[n,r,w]$-MDPC code, $w=\mathcal{O}(\sqrt{n})$~\cite{misoczki2013mdpc}. We denote $\Psi_{ldpc/mdpc}$ the $t$-error correcting algorithm of LDPC/MDPC code. The decoding of LDPC/MDPC codes is carried out via the original Bit Flipping Decoder~\cite{gallager1962low}, but several improvements appeared afterwards.

\subsection{QC-LDPC/QC-MDPC code construction}
The construction of the [n,r,w]-QC-LDPC/QC-MDPC code is based on the construction of the parity check matrix $H$ of length $n=rn_0$ and row-weight $w$. There are several techniques to construct $H$ including the "circulants row" technique where the matrix $H$ is formed by $n_0$ ($r \times r$)-circulant matrices $H=[H_0 H_1 .. H_{n_0-1}]$ such that $H_{n_0-1}$ is non-singular and $w= \sum \limits_{\underset{}{i=0}}^{n_0-1} w_i$. The corresponding generator matrix G has the form:
\begin{equation}
  G= 
  \left[ \begin{array}{@{}c|c@{}}
  I_{(n-r) \times (n-r)} &
   \begin{matrix}
      &  & (H_{n_0-1}^{-1} H_0)^T\\
      &  & (H_{n_0-1}^{-1} H_1)^T\\
      &  &\vdots\\
      &  & (H_{n_0-1}^{-1} H_{n_0-2})^T\\
   \end{matrix} \\
      
\end{array} \right]
  \label{eq:qcgen}
\end{equation}

\subsection{Error correcting capability of  LDPC code}
Due to the fact that LDPC codes are characterized by their sparse parity check matrix, they can achieve very good error correction capability, particularly when the decoder deals with soft-information inputs; in this case, the decoder is referred to as a soft-information BF decoder~\cite{richardson2001capacity}.
The McEliece cryptosystem has binary inputs and no soft-information inputs; thus, the decoder is known as a hard-decision BF decoder~\cite{gallager1962low}.
It is, however, impossible to achieve a zero Decoding Failure Rate (DFR).
The LDPC decoder is an iterative decoder, unlike to the bounded distance decoder.
The error correction capability of the LDPC decoder is determined by numerical simulation estimation of the DFR~\cite{xiao2007estimation, xiao2009error}, and no efficient algorithm can provide the exact value of the error correction capability of LDPC codes. 

\subsection{McEliece cryptosystem}
The McEliece cryptosystem~\cite{mceliece1978public} is a public-key cryptosystem based on the hardness of distinguishing a t-error correcting binary Goppa code $C$ whose the generator matrix $G \in \mathbb{F}_q^{k \times n}$ from a random code, and the hardness of decoding a random code.
\begin{itemize}
\item The public key ($G'=SGP$, $t$), such that $S \in \mathbb{F}_2^{k \times k}$ denotes an invertible scrambling matrix, and $P \in \mathbb{F}_2^{n \times n}$ represents a permutation matrix.
\item The secret key ($S,G,P$).
\end{itemize}
\subsubsection{Encryption}
Using the public key ($G'$,$t$), the transmitter encrypts the plaintext $x \in \mathbb{F}_2^k$ into the ciphertext $c \in \mathbb{F}_2^n$ as follows:\\
 $c=xG'+z$, such that $z \in \mathbb{F}_2^n$ and $wt(z) \leq t$.
\subsubsection{Decryption}
Using the t-error correcting algorithm $\Psi_{Goppa}$ of $C$, the receiver deciphers $c$ as follows:
\begin{enumerate}
\item $c'=\Psi_{Goppa}(cP^{-1})=mSG$, which is a valid codeword, then extract $m'=mS$ from $c'$.
\item $m=m'S^{-1}=(mS)S^{-1}$.
\end{enumerate}

\section{The proposed cryptosystem}
\label{secProp}
\subsection{Keys generation}
\begin{enumerate}
    \item  Choose an [$n,r,w_1$] QC-MDPC code $C_1$ of length $n$ and co-dimension $r$ that can correct $t_1$ errors, then from its parity check matrix $H_1$ generate its systematic generator matrix $G_1$, such that $r$ must be odd.
    \item  Choose an [$n,r,w_2$] QC-LDPC code $C_2$ of length $n$ and co-dimension $r$ that can correct $t_2$ errors, then from its parity check matrix $H_2$  generate its systematic generator matrix $G_2$.
    \item  Compute $G=\begin{bmatrix}
G_1 & G_1 \\
0 & G_2
\end{bmatrix}$, such that $G_1$ and $G_2$ have the same length $n=n_0r$ and the same row rank $k=n-r=k_0r$.
	\item Generate an invertible scrambling matrix $S \in \mathbb{F}_2^{k \times k}$, such that $S$ must be a dense matrix and it is formed by $(k_0 \times k_0)$ circulant blocks and each block has a size $r$.
	\item Compute $S'=\begin{bmatrix}
S & 0 \\
0 & S
\end{bmatrix}$.
	\item Compute 
	$G'=S'G=
\left[ \begin{array}{@{}c|c@{}}
  \begin{matrix}
SG_1 \\
0 
\end{matrix} &
\begin{matrix}
SG_1 \\
SG_2 
\end{matrix}  \\
\end{array} \right]	
	$ 
\end{enumerate}
\paragraph{The public key}  $\{G', t_1, t_2\}$.
\paragraph{The secret key}  $\{S',H_1, H_2\}$
\begin{remark}
For the storage of the public key $G'$ we can store just the right side of $G'$.
\end{remark}
 \subsection{Encryption}
 \label{sec:12}
 \begin{enumerate}
 \item Generate a binary messages $m=[m_1|m_2]$, such that $m_1 \in \mathbb{F}_2^k$ and $m_2 \in \mathbb{F}_2^k $.
 \item Randomly, generate a binary error vector $z=[z_1|z_2] \in \mathbb{F}_2^{2n}$ of weight $t_1$ in the left $n$ positions and $t_2$ in the right $n$ positions.
 \item $c= mG'+z+[(0...0)|h(z_1)]$ where $h(.)$ is a hash function and $h(z_1) \in \mathbb{F}_2^{n}$.
 \end{enumerate}
 
 \subsection{Decryption}
 \label{sec:13}
  \begin{enumerate}
 \item When the receiver receives the ciphertext $c$, he decomposes it into two blocks of same length $c=[c_1|c_2]$. 
 \item Using the decoding algorithm~\cite{drucker2019constant} we find
$c_1'= \Psi_{ldpc/mdpc}(c_1)=m_1SG_1$ and $z_1$.
\item Compute $c_2+c_1'+h(z_1)=(m_1SG_1+m_2SG_2+z_2+h(z_1))+m_1SG_1+h(z_1)=m_2SG_2+z_2$, then using the decoding algorithm~\cite{drucker2019constant} we find
$c_2'= \Psi_{ldpc/mdpc}(c_2+c_1'+h(z_1))=m_2SG_2$.
 \item Then extract $m_1'=m_1S$ from $c_1'$ and compute $m_1=m_1'S^{-1}$.
 \item Then extract $m_2'=m_2S$ from $c_2'$ and compute $m_2=m_2'S^{-1}$.
 \item Then, combining the two results $m=[m_1|m_2]$.
 \end{enumerate}

\section{Security analysis}
\label{secSecAnalysis}
\subsection{key recovery attacks}
\subsubsection{Weak key attack}
The attacker can perform a key recovery attack to obtain an equivalent parity check matrix $H_2'$ with row-weight $w_2$. The complexity of such attack depends on the row-weight $w_2$ and the number of rows in $H_2$ with row-weight $w_2$ as shown in~\cite{misoczki2013mdpc,dragoi2017cryptanalysis}. The attacker uses the public $SG_2$ matrix to derive the corresponding parity check matrix $H_2^*$, which is a valid parity check matrix but not a low weight matrix, as follows: 
\begin{enumerate}
    \item Compute $G_{2}^*=\left[ \begin{array}{@{}c|c@{}}
  I_{k \times k} & A
\end{array} \right]$ the systematic form of $SG_2$. 
\item Compute the corresponding parity check matrix $H_2^*=\left[ \begin{array}{@{}c|c@{}}
  A^T & I_{(n-k) \times (n-k)} 
\end{array} \right]$
\end{enumerate}
\begin{proposition}
The parity check matrix $H'=\left( \begin{matrix}
H_1^* & 0 \\
H_2^* & H_2^*
\end{matrix}\right)$ is a valid parity check matrix of the code generated by $G'$, such that $H_1^*$ is the parity check matrix of the QC-MDPC code $C_1$ and $H_2^*$ is the parity check matrix of the QC-LDPC code $C_2$, however $H_1^*$ and $H_2^*$ are not a low weight parity-check matrices.
\end{proposition}
\begin{proof}
\begin{center}
$G'{H'}^T=\left(
  \begin{matrix}
SG_1 & SG_1\\
0     & SG_2
\end{matrix} 
 \right) {\left( \begin{matrix}
H_1^* &  0 \\
H_2^* &  H_2^*
\end{matrix}\right)}^T=0$.
\end{center}
\end{proof}

Since the parity check matrix $H_2$ has $r$ rows of weight $w_2$, the attacker must employ one of the Information Set Decoding algorithms~\cite{lee1988observation,leon1988probabilistic,stern1988method,dumer1991minimum,becker2012decoding} to determine $r$ codewords of weight $w_2$ in the dual code of $C_2$. 
\par
In~\cite{sendrier2011decoding}, the author introduced a technique called "Decoding One Out of Many (DOOM)" which aims to gain a factor $g=(r/\sqrt{N_i})$ when $N_i$ instances are treated simultaneously and the system has $r$ solutions, but only one solution is sufficient. In our proposal, the target code $C_2$ is a quasi-cyclic code of co-dimension $r$. The attacker applies the ISD attack to find one codeword $v$ of weight $w_2$ in $C_2^\perp$ (i.e. $H_2^* v^T=0$); the time complexity to execute this type of attack is $O(2^{-w_2log_2(r/n)(1+o(1))})$ when $n \to \infty$~\cite{canto2016analysis}. The DOOM technique provides a gain $g=(r/\sqrt{1})$~\cite{misoczki2013mdpc}, so the total work factor is $WF_{rec}(n, k, w_2)=\frac{Cost(ISD(n,k,w_2))}{r}$. In the sequel, we will show that even if an attacker finds an  equivalent parity-check matrix for the QC-LDPC code, he can not decode $c_2$.
\par
It is worth noting that the attacker must perform both a reaction attack and a weak key attack in order to recover the whole key.
\subsubsection{Reaction attack}
QC-MDPC codes are at the heart of the BIKE submission~\cite{Abb}, a key encapsulation mechanism (KEM) that also reached the 4th round of the NIST competition for post quantum cryptography. It is true that several recent papers address the security of the QC-MDPC encryption schemes : In 2016, Guo, Johansson and Stankovski presented a key recovery (reaction) attack~\cite{guo2016key, guo2018key}, using the fact that the scheme is imperfectly correct. They identify a dependence between the secret key and the failure in decryption, and use that in the attack to build a distance spectrum for the secret key. The latest set is then used to build the secret key. The interesting fact for this attack is that it could be applied on the CCA version of the QC-MDPC scheme and still be successful. These type of attacks also applies on LDPC encryption schemes~\cite{fabvsivc2017reaction}, and open the question of improving the decoding algorithm for LDPC/MDPC codes in order to considerably reduce the decryption failure probability in QC-MDPC or QC-LDPC schemes as it is the only way for a countermeasure. So as we are talking about decryption of QC-LDPC/QC-MDPC in this scheme, it is also vulnerable to such attacks and the countermeasure is not in the design of the scheme, but in the decoding algoritnm that will be used for the QC-MDPC code. So the scheme could be secure assuming that there is a decoding algorithm for QC-MDPC that has a very low Decoding Failure Rate (DFR$ \leq 2^{-\lambda}$ for a security level $\lambda$~\cite{baldi2020analysis}), and a DFR$\approx 10^{-7}$ for the QC-LDPC which is adequate  for a practical purpose.

However, several recent works~\cite{sendrier2019decoding, drucker2019constant, sendrier2019low, drucker2020qc} then follow this direction in other to improve the security of MDPC schemes, specially with CCA security. Note that a first timing attack was also proposed in~\cite{eaton2018qc}. A recent class of weak keys for the QC-MDPC encryption scheme is also studied in~\cite{aydin2020class}.
The authors of~\cite{sendrier2019low} have demonstrated that by increasing the block size $r$ slightly (by $5\%$ to $10\%$ ) while keeping $w$ and $t$ and implementing a decoder with a negligible DFR, it is possible to achieve IND-CCA security in the KEM based on ($2r, r, w, t$)-QC-MDPC code. 
\par
The key security in our proposal is based on the difficulty of recovering the QC-MDPC code's parity-check matrix.
For a security level $\lambda$, a decoder with DFR$\leq 2^{-\lambda}$ is required~\cite{baldi2020analysis}; this is achieved by using the Backflip decoder~\cite{drucker2019constant} described  in BIKE-CCA KEM~\cite{Abb}. 

\subsection{Message recovery attacks}
\subsubsection{Decoding attack based on a weak key }
In~\cite{dragoi2017cryptanalysis} the authors showed that the security of the plaintext in the cryptosystem described in~\cite{moufek2016new} is based on the cost of finding $r$ vectors with low weights in the dual of the public QC-LDPC code. However, in our proposal this cannot happen even if an adversary successes to construct an equivalent parity check matrix of the private QC-LDPC code. Since the dual code of $C_2$ has many low-weight vectors, the adversary can find out its sparse parity check matrix $H_2$ using the attack described in~\cite{dragoi2017cryptanalysis}; but he cannot decode $c_2=m_1SG_1+m_2SG_2+z_2+h(z_1)$ since the expected number of errors is more than $t_2$. The adversary can also try to find $m_2$ by computing $c_3=c_1+c_2$ and apply the decoding algorithm of the QC-LDPC code to $c_3=m_2SG_2+z_2+z_1+h(z_1)$; however, the process will fail since the weight of $z_2+z_1+h(z_1)$ is, with a very high probability, large compared to the QC-LDPC correction capability assuming that the hash function $h$ is secure. Also remark that, in case the hash function exceptionally outputs a low weight vector $h(z_1)$, the weight of $z_2+z_1+h(z_1)$ will always be large compared to the QC-LDPC correction capability.  

\subsubsection{Information set decoding attack}
If $k$ and $n$ are the code dimension and the code length of the public code respectively, then knowing $k$ error-free bits from the ciphertext $c$ is enough to decrypt $c$ easily. However, finding out $k$ error-free bits is not easy and this is equivalent to decode a random code which turned out to be NP-complete~\cite{berlekamp1978inherent}. The most efficient ISD attack is executed in exponential time. In  this paper,  we  consider  the BJMM attack~\cite{becker2012decoding} to analyze the security level of our cryptosystem. In the case of our proposal, the adversary recovers the plaintext if he successfully decodes the QC-MDPC i.e. recovers $z_1$. Since the code is quasi-cyclic the required work factor can be minimized to $\frac{Cost(ISD_{BJMM})}{\sqrt{r}}$~\cite{sendrier2011decoding}. The total work factor $WF=\frac{Cost(ISD_{BJMM}(n,k,t1))}{\sqrt{r}}$. If we consider the suggested parameters~\cite{sendrier2019low} $(n=23558, k=11779, n_0=2, w_1=142, t_1=134)$ for the QC-MDPC code, then $WF=2^{128.9}$. Achieving this security level using these suggested parameters needs to store only $2(n_0-1).n_0.r$ bits$\approx 5.751$ Kbytes. Using some CCA2-secure conversions allows to store only $2r$ bits $\approx 2.876$ Kbytes. The security level depends solely of the cost of decoding the QC-MDPC regardless of the QC-LDPC; \tablename~\ref{codeParameters} shows the recommended code parameters for the three commonly used security levels (128-bit, 192-bit, and 256-bit) IND-CCA and IND-CPA variants. The following QC-MDPC code parameters~\cite{sendrier2019low} are recommended for BIKE submission~\cite{Abb}. For the ($n, r, w_2$)-QC-LDPC, $t_2$ must be chosen so that the DFR$\approx 2^{-7}$.

\begin{table}[htbp]\small
\caption{Our proposal's security exponent, the recommended QC-MDPC code parameters~\cite{sendrier2019low}, and the recommended QC-LDPC code parameters.}
\label{codeParameters}
\begin{center}
\begin{tabular}{ |c|c|c| } 
\hline
Security level & ($n, r, w_1, t_1$)-QC-MDPC & ($n, r, w_2$)-QC-LDPC\\
\hline
128-bit IND-CCA & $(23558, 11779, 142, 134)$ & $(23558, 11779, 14)$ \\ 
128-bit IND-CPA & $(20326, 10163, 142, 134)$ & $(20326, 10163, 14)$  \\ 
\hline
192-bit IND-CCA & $(49642, 24821, 206, 199)$ & $(49642, 24821, 15)$ \\ 
192-bit IND-CPA & $(39706, 19853, 206, 199)$ & $(39706, 19853, 15)$  \\ 
\hline
256-bit IND-CCA & $(81194, 40597, 274, 264)$ & $(81194, 40597, 15)$ \\ 
256-bit IND-CPA & $(65498, 32749, 274, 264)$ & $(65498, 32749, 15)$  \\ 
\hline
\end{tabular}
\end{center}
\end{table}

 \section{CONCLUSION}
 \label{conc}
 In this work, we have proposed a new McEliece cryptosystem scheme based on QC-MDPC and QC-LDPC codes using $(U|U+V)$ construction. We have proven that it is possible to use LDPC codes in the McEliece cryptosystem thanks to the $(U|U+V)$ construction. We have shown that our proposal is secure even if the QC-LDPC's parity check matrix is known, and even if weak keys attacks and reaction attacks are applied.


%




\bibliographystyle{fundam}
\bibliography{citations}

\end{document}